\documentclass[10pt,a4paper]{article}
\usepackage{fullpage}
\usepackage{amsfonts,amsmath,amssymb}
\usepackage{amsthm}
\usepackage{graphicx}
\usepackage{sectsty}
\usepackage{authblk}
\usepackage{empheq}
\usepackage[numbers, sort&compress]{natbib}

\theoremstyle{plain}
\newtheorem{theorem}{Theorem}[section]
\newtheorem{lemma}[theorem]{Lemma}
\newtheorem{proposition}[theorem]{Proposition}

\theoremstyle{definition}
\newtheorem{definition}[theorem]{Definition}
\theoremstyle{remark}

\numberwithin{equation}{section}

\sectionfont{\large}
\newenvironment{acknowledgement}[1][Acknowledgement
]{\begin{trivlist} \item[\hskip \labelsep {\bfseries
#1}]}{\end{trivlist}}

\begin{document}
\title{Topology in One Point Interaction Problem on Extended Non-Local Star Graphs and its Eigenvalues}
\author[1]{Lung-Hui Chen}
\affil[1]{\footnotesize General Education Center, Ming Chi University of Technology, New Taipei City, 24301, Taiwan.}
\maketitle
\begin{abstract}
The author studies the inverse spectral problem of Sturm-Liouville operator on a star-like metric graph. At the vertex of this star-like graph, there are attached $m$ edges that imposed with non-local  Sturm-Liouville operator satisfying some suitable non-local boundary conditions. At the vertex, we consider one point interaction condition to model a metric graph that fixed on the end of edges of the graph. This models the vibration or flux that changes over time that monitored at the vertex which serves as certain control/regulation center. The author shows that the system is solvable under very necessary conditions. It is crucial to recover  the topology of the network/metric graph which the topology is given. To begin the analysis, one constructs the special solution fixed on one end of edges while maintaining continuous at the vertex. This models a string that is vibrating vertically at the vertex according to certain frequencies. The non-local characteristic function plays a role, and then, one tries to find a non-trivial non-local eigenvalue.
\end{abstract}

\section{Introduction}
This paper investigates inverse spectral problems for Sturm-Liouville operators on star-shaped metric graphs. A star graph consists of several finite edges joined at a common central vertex and provides a natural mathematical model for wave vibration and transport phenomena in branched network structures. In recent years, Sturm-Liouville operators on metric graphs have received considerable attention because of roles in quantum graph theory, applications to telecommunications, electrical power networks, and transportation systems that connected to networks analysis. In particular, wave vibration and diffusion on such networks can be described by differential operators acting along on the edges, together with a few types of local or non-local boundary conditions and edge potentials. At the central vertex, some types of constraints like frozen-argument condition could be imposed to model a centralized coupling mechanism. 
Further more, the are a few types of local or non-local boundary conditions to construct the non-trivial eigenfunctions.
The non-local boundary conditions are not merely physical features, but also contribute to realizations of certain engineering and physical constraints. In particular, they help explaining the physical principles such as the conservation of energy in form of generalized Kirchhoff's law which express the standard continuity and flux conservation conditions that in turn described by the mathematical uniqueness of solutions of the state. The investigation presented here demonstrates that, under appropriate local or non-local conditions, the inverse spectral problem admits an inverse uniqueness and solvable formulation.
\par
The most interesting aspect of such modeling is its connection to the theory of metric quantum graphs, which studies the physical dynamics of single particle moving on metric graphs governed by various sorts of differential operators \cite{Alb,Bon,Boman,Exner,Freiling,Gerasimenko,Kottos,Kuchment2,Naimark,Nizhnik1,Nizhnik2} to name a few, along with many different kinds of boundary conditions, and often requires the energy conservational laws or even certain commercial requirements in real world scenarios. That is, the boundary condition acts as an integral part of the modeling these specific problems. We are interested to formulate and to solve a few number of direct and inverse spectral problems and scattering problems on quantum graphs \cite{Alb,Alb2,Belishev,Bon,Boman,Carlson,Gerasimenko,Kottos,Kostrykin,Kuchment2,Kurasov1,Kurasov2,Naimark,Nizhnik1,Nizhnik2,Pivovarchik1,Pivovarchik2,Pokornyi,Yurko1,Yurko2} to name a few.
Specifically, in this paper, we study the inverse spectral problems for Schr\"{o}dinger operator with non-local potential function on star graph. Previously, the inverse spectral problem on  finite graph for the Sturm-Liouville problem with non-local potential or with various boundary value condition was considered in \cite{Alb,Alb2,Bon2,Bon,Kostrykin,Kurasov1,Kurasov2,Nizhnik1,Nizhnik2,
Pivovarchik1,Yang,Shieh2}.
\par
We begin with the following star graph $T$: The central vertex of $T$ is located at the origin $0\in\mathbb{R}^{2}$ and connected by 
$m$ edges, each sharing the origin as their common vertex, which we denoted as $v_{0}$.  Each edge $e_{j}$ has known length $l_{j}$, $j=1,2,\ldots,m$ with end points denoted as $v_{j}$, $j=1,2,\ldots,m$, counted anti-clockwise, and the angles between the edges are assumed to be acute and denoted as $\{\theta_{j}\}_{j=1}^{m}$.
\par
On $T$ and its edge $e_{j}$, there exists the function $\psi_{j}(x)\in W_{2}^{2}(0,l_{j})$ such that the function $\psi_{j}(x)$ satisfies the following eigenvalue problem with complex-valued non-local potentials $q_{j}(x)\in L^{2}(0,l_{j})$:
\begin{eqnarray}\label{1.1}
-\frac{d^{2}}{dx^{2}}\psi_{j}(x)+q_{j}(x)\psi_{j}(0)=\lambda\psi_{j}(x),\,0<x<l_{j},\,j=1,2,\ldots,m,
\end{eqnarray}
in which $m\geq3$, $\lambda=z^{2}\in\mathbb{C}$ is the spectral parameter,
and it is imposed with the boundary conditions
\begin{eqnarray}\label{1.2}
&&\psi_{j}(l_{j})=0,\,j=1,2,\ldots,m;\\
&&\psi_{1}(0)=\psi_{2}(0)=\cdots=\psi_{m}(0),\label{1.22}
\end{eqnarray}
and the non-local boundary value condition
\begin{equation}\label{1.3}
\sum_{j=1}^{m}\big\{\frac{d}{dx}\psi_{j}(0)-\int_{0}^{l_{j}}\psi_{j}(x)\overline{q_{j}(x)}dx\big\}=0,
\end{equation}
which serves a generalized Kirchhoff's law as a non-local conservation law, and the solution $\psi_{j}(x,z)$ depends on $(x,z)\in\mathbb{R}\times\mathbb{C}$ for each $j$.
We note that the inverse spectral problems is posed with a frozen argument at the vertex that is called a non-local point interaction in quantum mechanics, and  the generalized Kirchhoff's law reads as
\begin{equation}\label{1.4}
\sum_{j=1}^{m}\big\{\frac{d}{dx}\psi_{j}(0,z)-\int_{0}^{l_{j}}\psi_{j}(x,z)\overline{q_{j}(x)}dx\big\}=0,\,z\in\mathbb{C}.
\end{equation}
In this paper, we use the non-local spectral data to recover the information on angular information between the edges $$\{e_{1},e_{2},\cdots,e_{m}\}.$$
To this end, we consider the extended closed graph from the star graph $T$, which is denoted as $\overline{T}$ obtained by connecting the end point $v_{j}$, $j=1,2,\ldots,m$, and set $v_{m+1}=v_{1}$. Let us denote the angle between $e_{j}$ and $e_{j+1}$ as $\theta_{j}$, and the angle between $e_{m}$ and $e_{1}$ as $\theta_{m}$. Moreover, we call the new edges between the vertices $v_{j}$ and $v_{j+1}$ as $\bar{e}_{j}$ with length $\bar{l}_{j}$, and the edge between $v_{m}$ and $v_{1}$ as $\bar{e}_{m}$ with length $\bar{l}_{m}$. 
We set $\cup_{j=1}^{m}\bar{e}_{j}$ to be a cyclic metric graph by taking $v_{m+1}=v_{1}$ that moving around $v_{0}$.
Surely, $v_{j},\,j\geq1,$ is a vertex attached with edge $e_{j}$, $\bar{e}_{j-1}$, and $\bar{e}_{j+1}$. The extended graph $\overline{T}$ is trivially unique.
\par
To begin with the analysis on $\overline{T}$, we modify the differential system~(\ref{1.1}),~(\ref{1.2}),~(\ref{1.22}), and~(\ref{1.4}) for the presence of  $\bar{e}_{j}$, $j=1,2,\cdots,m.$ For $z\in\mathbb{C}$,
\begin{empheq}[left=\text{$(\mathcal{D})$ }\empheqlbrace]{align}
-\frac{d^{2}}{dx^{2}}\psi_{j}(x,z)+q_{j}(x)\psi_{j}(0,z)=z^{2}\psi_{j}(x,z),\,x\in e_{j},\,j=1,2,\ldots,m;\tag{1.6a}\label{1.6a}
\\\tag{1.6b}\label{1.6b}
-\frac{d^{2}}{dx^{2}}\bar{\psi}_{j}(x,z)=z^{2}\bar{\psi}_{j}(x,z),\,x\in \bar{e}_{j},\,j=1,2,\ldots,m;\\
\tag{1.6c}\label{1.6d}
\psi_{1}(0,z)=\psi_{2}(0,z)=\cdots=\psi_{m}(0,z);\\
\tag{1.6d}\label{1.6e}
\psi_{j}(l_{j}^{-},z)=0,\,j=1,2,\cdots,m;\\
\nonumber
\sum_{j=1}^{m}\frac{d}{dx}\psi_{j}(0^{+},z)+\frac{d}{dx}\psi_{j}(l_{j}^{-},z)+\frac{d}{dx}\bar{\psi}_{j-1}(\bar{l}_{j-1}^{-},z)+\frac{d}{dx}\bar{\psi}_{j}(0^{+},z)\\\hspace{80pt}-\int_{0}^{l_{j}}\psi_{j}(x,z)\overline{q_{j}(x)}dx=0,\tag{1.6e}\label{1.6f}
\end{empheq}
in which the Dirichlet problem~(\ref{1.6b}) and~(\ref{1.6e}) on the extended edge $\bar{e}_{j}$ defines the special solution $\bar{\psi}_{j}(x,z)$  as $z$ regarded as a fixed parameter in $\mathbb{C}$. Moreover, in~(\ref{1.6f}),$$\frac{d}{dx}\psi_{j}(l_{j}^{-},z)+\frac{d}{dx}\bar{\psi}_{j-1}(\bar{l}_{j-1}^{-},z)+\frac{d}{dx}\bar{\psi}_{j}(0^{+},z)$$ prepares for non-local Kirchhoff's condition at vertex $v_{j}$ on $\cup_{j=1}^{m}\bar{e}_{j}$, and the equation~(\ref{1.6f}) is the non-local Kirchhoff's condition in the system $(\mathcal{D})$ which is an entire function of finite type that its zeros define the set of eigenvalues. Thus, one needs to elaborate on the connection between the zero set of~(\ref{1.6b}),~(\ref{1.6e}) and~(\ref{1.6f}). The focus is the analysis of discrete set outside the zeros of $$\prod_{k=1}^{m}\sin{z l_{k}}\prod_{k=1}^{m} \sin{z \bar{l}_{k}}.$$
\par
To retrieve the topological information on $T$, we aim to recover the information on $\theta_{j}$, $j=1,2,\ldots,m$, with a knowledge of $q_{j}(x)$ on $L^{2}(0,l_{j})$, $j=1,2,\ldots,m$.

\section{Compatibility and Characteristic Functions}
In this section, we construct the special solutions of system $(\mathcal{D})$.
We start with the Fourier-sine series in $L^{2}(0,l_{j})$:
$$q_{j}(x)=\sum_{n=1}^{\infty}q_{j,n}\sin[\frac{n\pi}{l_{j}} (l_{j}-x)]$$ with Fourier-sine transform coefficients
\begin{eqnarray}\label{2.1}
&&\mathcal{F}(q_{j})(z)=\frac{2}{l_{j}}\int_{0}^{l_{j}}q_{j}(x)\sin[\frac{z\pi}{l_{j}} (l_{j}-x)]dx;\\
&&q_{j,n}:=\frac{2}{l_{j}}\int_{0}^{l_{j}}q_{j}(x)\sin[\frac{n\pi}{l_{j}} (l_{j}-x)]dx.
\end{eqnarray}
Firstly, we consider the special solutions of system $(\mathcal{D})$ with $\lambda=z^{2}$. Modifying Nizhnik's construction of special solutions as in \cite{Alb,Alb2,Nizhnik1,Nizhnik2}, we choose
\begin{equation} \label{2.2}
\phi_{j}(x;z)=\Big(\sin{z(l_{j}-x)}+\sin{z l_{j}}\sum_{n=1}^{\infty}\frac{q_{j,n}\sin[\frac{n\pi}{l_{j}}(l_{j}-x)]}{z^{2}-(\frac{n\pi}{l_{j}})^{2}}\Big)\prod_{k\neq j}\sin{z l_{k}}\prod_{k=1}^{m}\sin{z \bar{l}_{k}},\,x\in T,
\end{equation}
and then verify that
\begin{eqnarray}\label{2.4}
&&\phi_{1}(0;z)=\phi_{2}(0;z)=\cdots=\phi_{m}(0;z)=\prod_{k=1}^{m}\sin{z l_{k}}\prod_{k=1}^{m} \sin{z \bar{l}_{k}};
\\&&\label{2.5}
\phi_{1}(l_{1};z)=\phi_{2}(l_{2};z)=\cdots=\phi_{m}(l_{m};z)=0.
\end{eqnarray}
Thus, one point interaction at $v_{0}=0$ in~(\ref{1.6d}) is satisfied by~(\ref{2.4}) for $z\in\mathbb{C}$.
Moreover,
\begin{equation}\label{3333}
\phi_{j}'(x;z)=\Big(-z\cos{z(l_{j}-x)}-\sin{z l_{j}}\sum_{n=1}^{\infty}\frac{n\pi}{l_{j}}\frac{q_{j,n}\cos[\frac{n\pi}{l_{j}}(l_{j}-x)]}{z^{2}-(\frac{n\pi}{l_{j}})^{2}}\Big)\prod_{k\neq j}\sin{z l_{k}}\prod_{k=1}^{m}\sin{z \bar{l}_{k=1}},
\end{equation} 
and then
\begin{eqnarray}\label{2.3}
\phi_{j}'(0;z)=\Big(-z\cos{zl_{j}}-\sin{z l_{j}}\sum_{n=1}^{\infty}\frac{(-1)^{n}n\pi}{l_{j}}\frac{q_{j,n}}{z^{2}-(\frac{n\pi}{l_{j}})^{2}}\Big)\prod_{k\neq j}\sin{z l_{k}}\prod_{k=1}^{m}\sin{z \bar{l}_{k=1}},\,j=1,2,\ldots,m;
\end{eqnarray}
\begin{equation}
\phi_{j}'(l_{j};z)=\Big(-z-\sin{z l_{j}}\sum_{n=1}^{\infty}\frac{(-1)^{n}n\pi}{l_{j}}\frac{q_{j,n}}{z^{2}-(\frac{n\pi}{l_{j}})^{2}}\Big)\prod_{k\neq j}\sin{z l_{k}}\prod_{k=1}^{m}\sin{z \bar{l}_{k=1}}.
\end{equation} 

To prepare Kirchhoff's condition, we sum over $j$-th summands.
\begin{eqnarray}
\sum_{j=1}^{m}\phi_{j}'(0;z)=-\sum_{j=1}^{m}\Big(z\cos{zl_{j}}+\sin{z l_{j}}\sum_{n=1}^{\infty}\frac{(-1)^{n}n\pi}{l_{j}}\frac{q_{j,n}}{z^{2}-(\frac{n\pi}{l_{j}})^{2}}\Big)\prod_{k\neq j}\sin{z l_{k}}\prod_{k=1}^{m}\sin{z \bar{l}_{k=1}},\label{25}
\end{eqnarray}
and
\begin{equation}
\sum_{j=1}^{m}\phi_{j}'(l_{j};z)=-\sum_{j=1}^{m}\Big(z+\sin{z l_{j}}\sum_{n=1}^{\infty}\frac{(-1)^{n}n\pi}{l_{j}}\frac{q_{j,n}}{z^{2}-(\frac{n\pi}{l_{j}})^{2}}\Big)\prod_{k\neq j}\sin{z l_{k}}\prod_{k=1}^{m}\sin{z \bar{l}_{k=1}}.
\end{equation} 
On the presence of $\{\bar{e}_{j}\}_{j=1}^{m}$ and $\{v_{j}\}_{j=1}^{m}$ on $\overline{T}$, we choose the special solutions on each $\bar{e}_{j}$:
\begin{equation}
\bar{\phi}_{j}(x;z)=\sin{z(\bar{l}_{j}-x)}\prod_{k=1}^{m}\sin{z l_{k}}\prod_{k\neq j}\sin{z \bar{l}_{k=1}},\,0\leq x\leq\bar{l}_{j},
\end{equation}
and
\begin{equation}
\bar{\phi}_{j}'(x;z)=-z\cos{z(\bar{l}_{j}-x)}\prod_{k=1}^{m}\sin{z l_{k}}\prod_{k\neq j}\sin{z \bar{l}_{k=1}},\,0\leq x\leq\bar{l}_{j}.\label{555}
\end{equation}
To prepare the generalized Kirchhoff condition, we use flux~(\ref{3333}) and~(\ref{555}) at vertex $v_{j}$, $j=1,2,\cdots,m,$
to deduce that
\begin{eqnarray}\label{2133}\nonumber
&&\hspace{12pt}\frac{d}{dx}\phi_{j}(l_{j}^{-},z)+\frac{d}{dx}\bar{\phi}_{j-1}(\bar{l}_{j-1}^{-},z)+\frac{d}{dx}\bar{\phi}_{j}(0^{+},z)\\
&&\nonumber=-\Big(z+\sin{z l_{j}}\sum_{n=1}^{\infty}\frac{(-1)^{n}n\pi}{l_{j}}\frac{q_{j,n}}{z^{2}-(\frac{n\pi}{l_{j}})^{2}}\Big)\prod_{k\neq j}\sin{z l_{k}}\prod_{k=1}^{m}\sin{z \bar{l}_{k}}\\
&&-\Big(z+z\cos{z\bar{l}_{j}}\Big)\prod_{k=1}^{m}\sin{z l_{k}}\prod_{k\neq j}\sin{z \bar{l}_{k=1}}.\label{288}
\end{eqnarray}
Secondly, we recall the non-local condition in~(\ref{1.6f}) on $T$,
\begin{eqnarray}\nonumber
&&\int_{0}^{l_{j}}\phi_{j}(x;z)\overline{q_{j}(x)}dx\\&=&\Big(\int_{0}^{l_{j}}\sin{z(l_{j}-x)}\overline{q_{j}(x)}dx+\sin{z l_{j}}\sum_{n=1}^{\infty}\frac{q_{j,n}\int_{0}^{l_{j}}\sin[\frac{n\pi}{l_{j}}(l_{j}-x)]\overline{q_{j}(x)}dx}{z^{2}-(\frac{n\pi}{l_{j}})^{2}}\Big)\prod_{k\neq j}\sin{z l_{k}}\prod_{k=1}^{m}\sin{z \bar{l}_{k}}.\nonumber
\end{eqnarray}
Then, on the edges $\{e_{j}\}$, we consider the summation
\begin{eqnarray}\label{29}
&&\sum_{j=1}^{m}\int_{0}^{l_{j}}\phi_{j}(x;z)\overline{q_{j}(x)}dx\\\hspace{-5pt}&=&\hspace{-5pt}\sum_{j=1}^{m}\Big(\int_{0}^{l_{j}}\sin{z(l_{j}-x)}\overline{q_{j}(x)}dx+\sin{z l_{j}}\sum_{n=1}^{\infty}\frac{q_{j,n}\int_{0}^{l_{j}}\sin[\frac{n\pi}{l_{j}}(l_{j}-x)]\overline{q_{j}(x)}dx}{z^{2}-(\frac{n\pi}{l_{j}})^{2}}\Big)\prod_{k\neq j}\sin{z l_{k}}\prod_{k=1}^{m}\sin{z \bar{l}_{k}}.\nonumber\end{eqnarray}
Therefore, the equation~(\ref{25}),~(\ref{2133}), and~(\ref{29}) add up to be the characteristic function in this paper:
\begin{eqnarray}\nonumber
&&\Phi(z)\underset{def}{=}\sum_{j=1}^{m}\Big(\int_{0}^{l_{j}}\sin{z(l_{j}-x)}\overline{q_{j}(x)}dx+\sin{z l_{j}}\sum_{n=1}^{\infty}\frac{q_{j,n}\overline{q_{j,n}}}{z^{2}-(\frac{n\pi}{l_{j}})^{2}}\Big)\prod_{k\neq j}\sin{z l_{k}}\prod_{k=1}^{m}\sin{z \bar{l}_{k}}\\\nonumber
&&\nonumber+\sum_{j=1}^{m}\Big( 2z+z\cos{zl_{j}}+2\sin{z l_{j}}\sum_{n=1}^{\infty}\frac{(-1)^{n}n\pi}{l_{j}}\frac{q_{j,n}}{z^{2}-(\frac{n\pi}{l_{j}})^{2}}\Big)\prod_{k\neq j}\sin{z l_{k}}\prod_{k=1}^{m}\sin{z \bar{l}_{k}}\\
&&
+\sum_{j=1}^{m}\Big(z+z\cos{z\bar{l}_{j}}\Big)\prod_{k=1}^{m}\sin{z l_{k}}\prod_{k\neq j}\sin{z \bar{l}_{k=1}}.\label{Phi}
\end{eqnarray}

\section{Lemmata and Results}
We need to show the system $(\mathcal{D})$ is well-posed firstly, and the set of eigenvalues is not empty.
\begin{definition}
We say real numbers $\{a_{1},\,a_{2},\ldots, a_{m}\},\,m\geq 2$, are called rationally independent if the identity $ \sum_{j=1}^{m}n_{j} a_{j}=0$,  for $n_{j}\in\mathbb{Z}$, implies that $$n_{1}=n_{2}=\cdots=n_{m}=0.$$
\end{definition}
\begin{lemma}\label{3.2}
Let us assume that $\{l_{j},\bar{l}_{j}\}_{j=1}^{m}$ are mutually independent. Then, for  $z\in\frac{\mathbb{Z}\pi}{l_{j}},\,j=1,2,\ldots,m,$  $\phi_{j}(x;z)$ satisfies~(\ref{1.1}),~(\ref{1.2}), and~(\ref{1.22}).
\end{lemma}
\begin{proof}
It is sufficient to proceed with~(\ref{2.2}):
\begin{equation} \nonumber
\phi_{j}(x;z)=\Big(\sin{z(l_{j}-x)}+\sin{z l_{j}}\sum_{n=1}^{\infty}\frac{q_{j,n}\sin[\frac{n\pi}{l_{j}}(l_{j}-x)]}{z^{2}-(\frac{n\pi}{l_{j}})^{2}}\Big)\prod_{k\neq j}\sin{z l_{k}}\prod_{k=1}^{m}\sin{z \bar{l}_{k}},
\end{equation}
to obtain
\begin{equation}
-\phi''_{j}(x;z)=\Big(z^{2}\sin{z(l_{j}-x)}+(\frac{n\pi}{l_{j}})^{2}\sin{z l_{j}}\sum_{n=1}^{\infty}\frac{q_{j,n}\sin[\frac{n\pi}{l_{j}}(l_{j}-x)]}{z^{2}-(\frac{n\pi}{l_{j}})^{2}}\Big)\prod_{k\neq j}\sin{z l_{k}}\prod_{k=1}^{m}\sin{z \bar{l}_{k}},\label{331}
\end{equation}
and
\begin{eqnarray}\nonumber
q_{j}(x)\phi_{j}(0;z)&=&q_{j}(x)\Big(\sin{zl_{j}}+\sin{z l_{j}}\sum_{n=1}^{\infty}\frac{q_{j,n}\sin[n\pi]}{z^{2}-(\frac{n\pi}{l_{j}})^{2}}\Big)\prod_{k\neq j}\sin{z l_{k}}\prod_{k=1}^{m}\sin{z \bar{l}_{k}}\\&=&q_{j}(x)\prod_{k=1}^{m}\sin{z l_{k}}\prod_{k=1}^{m}\sin{z \bar{l}_{k}}.\label{332}
\end{eqnarray}
Hence, we deduce
\begin{equation}
(\frac{n\pi}{l_{j}})^{2}\phi_{j}(x;z)=(\frac{n\pi}{l_{j}})^{2}\Big(\sin{z(l_{j}-x)}+\sin{z l_{j}}\sum_{n=1}^{\infty}\frac{q_{j,n}\sin[\frac{n\pi}{l_{j}}(l_{j}-x)]}{z^{2}-(\frac{n\pi}{l_{j}})^{2}}\Big)\prod_{k\neq j}\sin{z l_{k}}\prod_{k=1}^{m}\sin{z \bar{l}_{k}},\label{333}
\end{equation}
and
\begin{equation}\label{0.3}
-\phi''_{j}(x;z)+q_{j}(x)\phi_{j}(0;z)=(\frac{n\pi}{l_{j}})^{2}\phi_{j}(x;z),\,z\in\frac{\mathbb{Z}\pi}{l_{j}},
\end{equation}
since $q_{j}(x)\phi_{j}(0;z)=0$ for $z\in\frac{\mathbb{Z}\pi}{l_{j}}$ by the rational independence assumption. Thus, the lemma is proven. 

\end{proof}
From~(\ref{0.3}), we realize that this is merely a degenerate case. We are interested to find out the eigenvalues of $(\mathcal{D})$ outside the zeros of 
$\prod_{k=1}^{m}\sin{z l_{k}}\prod_{k=1}^{m}\sin{z \bar{l}_{k}}$.
\par
For reader's convenience, we include the following classical lemma due to E. C. Titchmarsh \cite{Titchmarsh}, and a modern proof can be found in \cite{Tang}.
\begin{lemma}[Titchmarsh]\label{L2.1}
Let $u\in\mathcal{E}'(\mathbb{R})$, then
$$N_{\mathcal{F}(u)}(r)= \frac{|\mbox{ch supp } u|}{\pi}\big(r + o(1)\big)$$ and $$N_{f}(r)=\sum_{|z|\leq r}\frac{1}{2\pi}\oint_{z}\frac{f'(\omega)}{f(\omega)}d\omega,$$
and the phrase $''\mbox{ch supp }u''$ means the convex hull of the effective support of $u$. 
\par
Moreover, $N_{f}(r)$ is the counting function of the zeros of $f$ inside a ball of radius $r$. We count the zeros according to their multiplicities.
\end{lemma}
\begin{definition}
We denote the following quantity as the zero density of an entire function $f$ of finite type.
\begin{equation}
\delta(f):=\lim_{r\rightarrow\infty}\frac{N_{f}(r)}{r}.
\end{equation}
\end{definition}
\begin{lemma}\label{35}
Let $f$, $g$ be two completely regular growth entire functions of finite type with density functions $\delta(f)$ and
$\delta(g)$ respectively. Then, 
\begin{equation}\nonumber
\delta(fg)=\delta(f)+\delta(g),
\end{equation}
and
\begin{equation}
\delta(f+g)=\max\{\delta(f),\delta(g)\},\label{220}
\end{equation}if the types of two functions are not equal. 
\par
In case that $S$ is a discrete set in $\mathbb{C}$, we denote $\delta(S)$ as its zero density in $\mathbb{C}$ similarly.
\end{lemma}
\begin{proof}
We refer to B. Ya. Levin's books \cite[p.\,52]{Levin,Levin2} and Boas \cite{Boas} for detailed introduction.  It is all about completely regular growth functions of finite type.

\end{proof}
One of the most important entire functions is Fourier transform.
\begin{lemma}\label{LL33}
The following representation holds:
\begin{equation}
\sin{z l_{j}}\sum_{n=1}^{\infty}\frac{(-1)^{n}n\pi}{l_{j}}\frac{q_{j,n}}{z^{2}-(\frac{n\pi}{l_{j}})^{2}}=\frac{1}{2\pi}\int_{0}^{\pi}q_{j}(\frac{l_{j}}{\pi}(\pi-x))\sin{zx}dx.
\end{equation}
$q_{j,n}=0$, $n\in\mathbb{Z}$, necessarily and sufficiently that $q_{j}(x)\equiv0$, for  $0\leq j\leq m$.
\end{lemma}
\begin{proof}
We recall that 
$$q_{j}(x)=\sum_{n=1}^{\infty}q_{j,n}\sin[\frac{n\pi}{l_{j}} (l_{j}-x)]$$ with Fourier coefficients
$$q_{j,n}=\frac{2}{l_{j}}\int_{0}^{l_{j}}q_{j}(x)\sin[\frac{n\pi}{l_{j}} (l_{j}-x)]dx.$$
Using
\begin{eqnarray}
\frac{q_{j,n}}{z^{2}-(\frac{n\pi}{l_{j}})^{2}}=\frac{-\frac{1}{2}\frac{l_{j}}{n\pi}q_{j,n}}{(z+\frac{n\pi}{l_{j}})}+\frac{\frac{1}{2}\frac{l_{j}}{n\pi}q_{j,n}}{(z-\frac{n\pi}{l_{j}})},
\end{eqnarray}
we deduce
\begin{eqnarray}\nonumber
\sin{z l_{j}}\sum_{n=1}^{\infty}\frac{(-1)^{n}n\pi}{l_{j}}\frac{q_{j,n}}{z^{2}-(\frac{n\pi}{l_{j}})^{2}}&\hspace{-4pt}=\hspace{-4pt}&\sin{z l_{j}}\sum_{n=1}^{\infty}(-1)^{n}\frac{-\frac{1}{2}q_{j,n}}{(z+\frac{n\pi}{l_{j}})}+\sin{z l_{j}}\sum_{n=1}^{\infty}(-1)^{n}\frac{\frac{1}{2}q_{j,n}}{(z-\frac{n\pi}{l_{j}})}\\\nonumber
&\hspace{-4pt}=\hspace{-4pt}&\sin{z l_{j}}\sum_{n=-1}^{\infty}(-1)^{n}\frac{\frac{1}{2}q_{j,n}}{(z-\frac{n\pi}{l_{j}})}+\sin{z l_{j}}\sum_{n=1}^{\infty}(-1)^{n}\frac{\frac{1}{2}q_{j,n}}{(z-\frac{n\pi}{l_{j}})}
\\&\hspace{-4pt}=\hspace{-4pt}&\sin{z l_{j}}\sum_{n=-\infty}^{\infty}(-1)^{n}\frac{\frac{1}{2}q_{j,n}}{(z-\frac{n\pi}{l_{j}})}\nonumber
\\&\hspace{-4pt}=\hspace{-4pt}&\sin{\eta\pi}\sum_{n=-\infty}^{\infty}(-1)^{n}\frac{\frac{1}{2}q_{j,n}}{(\frac{\eta\pi}{l_{j}}-\frac{n\pi}{l_{j}})},\label{0.39}
\end{eqnarray}
in which $z:=\frac{\eta\pi}{l_{j}}$, $\eta\in\mathbb{C}$, and $x:=\frac{l_{j}}{\pi}y$. Following the change of variables, we deduce
\begin{equation}
\int_{0}^{\pi}q_{j}(\frac{l_{j}}{\pi}(\pi-y))\sin{(ny)}dy=\frac{\pi}{2}q_{j,n},
\end{equation}
that verifies~(\ref{2.1}). Moreover, we verify more for~(\ref{2.1}) that
\begin{eqnarray}
\frac{1}{\pi}\int_{0}^{\pi}q_{j}(\frac{l_{j}}{\pi}(\pi-x))\sin{zx}dx=\frac{1}{l_{j}}\int_{0}^{l_{j}}q_{j}(x)\sin{z(\frac{\pi}{l_{j}}(l_{j}-x))}dx.
\end{eqnarray}
Therefore,~(\ref{0.39}) becomes
\begin{eqnarray}\nonumber
&&\sin{z l_{j}}\sum_{n=1}^{\infty}\frac{(-1)^{n}n\pi}{l_{j}}\frac{q_{j,n}}{z^{2}-(\frac{n\pi}{l_{j}})^{2}}\\&\hspace{-4pt}=&\sin{\eta\pi}\sum_{n=-\infty}^{\infty}(-1)^{n}\frac{\frac{1}{\pi}\int_{0}^{\pi}q_{j}(\frac{l_{j}}{\pi}(\pi-y))\sin{ny}dy}{\frac{\pi}{l_{j}}(\eta-n)}\nonumber\\&\hspace{-4pt}=&\sin{\eta\pi}\sum_{n=-\infty}^{\infty}(-1)^{n}\frac{\frac{l_{j}}{\pi}\int_{0}^{\pi}q_{j}(\frac{l_{j}}{\pi}(\pi-y))\sin{ny}dy}{\pi(\eta-n)},\label{3100}
\end{eqnarray}
with Fourier coefficients
$$\frac{l_{j}}{\pi}\int_{0}^{\pi}q_{j}(\frac{l_{j}}{\pi}(\pi-y))\sin{ny}dy,$$
that defines a Fourier transform as we apply the interpolation theory in complex analysis referring to M. Cartwright and Levin \cite[p.\,150,\,151]{Levin2}, and
$\{\frac{l_{j}}{\pi}\int_{0}^{\pi}q(\frac{l_{j}}{\pi}(\pi-y))\sin{ny}dy\}_{n=1}^{\infty}\in l^{2}.$
A Fourier transform that vanishes on $\mathbb{Z}$ must be trivial \cite[p.\,150,\,Theorem\,1]{Levin2} which is known as Cartwright's theorem in the literature. The lemma is thus proven.

\end{proof}
\begin{lemma}\label{CC}
Let $f(z),\,z=x+iy\in\mathbb{C}$, be an entire function of finite type with type $\sigma\leq\pi$, such that $f(n) = 0$, for all $n\in\mathbb{Z}$, and 
$$\lim_{y\rightarrow\infty}f(iy)e^{-\pi|y|}=0.$$
Then, $f(z)=0$.
\end{lemma}
\begin{proof}
We refer the related proof and its terminology to \cite{Boas,Levin,Levin2}, and especially, to \cite[p.\.151,\,p.\,152]{Levin2}. 

\end{proof}
\begin{definition}\label{366}
We define
\begin{eqnarray}
&&\mathcal{Z}_{0}=\big\{z\in\mathbb{C}\big|\,\prod_{k=1}^{m}\sin{z \bar{l}_{k}\prod_{k=1}^{m}\sin{z l_{k}}}=0\big\};\\
&&\mathcal{Z}\,\,=\big\{z\in\mathbb{C}\big|\,\Phi(z)=0\big\};
\\&&
\mathcal{Z}_{0,j}=\big\{z\in\mathbb{C}\big|\,\sin{z  l_{j}}=0\big\};\\&&
\bar{\mathcal{Z}}_{0,j}=\big\{z\in\mathbb{C}\big|\,\sin{z  \bar{l}_{j}}=0\big\}.
\end{eqnarray}
\end{definition}
\begin{proposition}[Weyl's theorem]\label{377}
The following equality holds:
$\delta(\mathcal{Z}_{0})=\frac{\sum_{k=1}^{m}l_{k}+\bar{l}_{k}}{\pi}$.\end{proposition}
\begin{proof}
The result is a direct consequence of Lemma \ref{35} and Lemma \ref{LL33}. The highest density comes from the term $$\delta(z\cos{z\bar{l}_{j}}\prod_{k=1}^{m}\sin{z l_{k}}\prod_{k\neq j}\sin{z \bar{l}_{k}})$$ in $\Phi(z)$ with density $\frac{\sum_{k=1}^{m}l_{k}+\bar{l}_{k}}{\pi}$.

\end{proof}
We show that the elements in $\mathcal{Z}_{0}$ under the assumption of rational independence of $\{l_{j},\bar{l}_{j}\}_{j=1}^{m}$ justifies the compatibility of differential system $\mathcal{D}$. This is also the existence of the non-trivial non-local eigenvalue of system $\mathcal{D}$.
\begin{theorem}Let us assume that $\{l_{j},\bar{l}_{j}\}_{j=1}^{m}$ is mutually independent and $\{q_{j}(x)\}_{j=1}^{m}$ is non-trivial. Then, 
$\mathcal{Z}_{0}\not\subset\mathcal{Z}$. 
\end{theorem}
\begin{proof}
We begin with the Dirichlet problem on $\overline{T}$:
\begin{empheq}[left=\empheqlbrace]{align}
-\frac{d^{2}}{dx^{2}}\psi_{j}(x,z)+q_{j}(x)\psi_{j}(0,z)=z^{2}\psi_{j}(x,z),\,x\in e_{j},\,j=1,2,\ldots,m\nonumber;\\\nonumber
-\frac{d^{2}}{dx^{2}}\bar{\psi}_{j}(x,z)=z^{2}\bar{\psi}_{j}(x,z),\,x\in \bar{e}_{j},\,j=1,2,\ldots,m;\\\nonumber
\psi_{1}(0,z)=\psi_{2}(0,z)=\cdots=\psi_{m}(0,z);\\\nonumber
\psi_{j}(l_{j}^{-},z)=0,\,j=1,2,\cdots,m;\\
\nonumber
\Phi(z)=0,
\end{empheq}
which is well-posed on $\overline{T}$, and
\begin{eqnarray}\nonumber
&&\Phi(z)=\sum_{j=1}^{m}\Big(\int_{0}^{l_{j}}\sin{z(l_{j}-x)}\overline{q_{j}(x)}dx+\sin{z l_{j}}\sum_{n=1}^{\infty}\frac{q_{j,n}\overline{q_{j,n}}}{z^{2}-(\frac{n\pi}{l_{j}})^{2}}\Big)\prod_{k\neq j}\sin{z l_{k}}\prod_{k=1}^{m}\sin{z \bar{l}_{k}}\\\nonumber
&&\nonumber+\sum_{j=1}^{m}\Big( 2z+z\cos{zl_{j}}+2\sin{z l_{j}}\sum_{n=1}^{\infty}\frac{(-1)^{n}n\pi}{l_{j}}\frac{q_{j,n}}{z^{2}-(\frac{n\pi}{l_{j}})^{2}}\Big)\prod_{k\neq j}\sin{z l_{k}}\prod_{k=1}^{m}\sin{z \bar{l}_{k}}\\
&&
+\sum_{j=1}^{m}\Big(z+z\cos{z\bar{l}_{j}}\Big)\prod_{k=1}^{m}\sin{z l_{k}}\prod_{k\neq j}\sin{z \bar{l}_{k}}.\label{3314}
\end{eqnarray}
Let $z^{\ast}\in\bar{\mathcal{Z}}_{j}$, that is, the zeros of $\sin{z  \bar{l}_{j}}$ for $j=1,2,\cdots,m$.
Using~(\ref{3314}),
\begin{eqnarray}\nonumber
&&\Phi(z^\ast)=\Big(z^{\ast}+z^{\ast}\cos{z^{\ast}\bar{l}_{j}}\Big)\prod_{k=1}^{m}\sin{z^\ast l_{k}}\prod_{k\neq j}\sin{z^\ast \bar{l}_{k}},
\end{eqnarray}
due to the rational independence of $\{l_{j},\bar{l}_{j}\}_{j=1}^{m}$.
\par
Suppose that $\bar{\mathcal{Z}}_{0j}$ is a proper subset of $\mathcal{Z}$ and look for contradiction. That is,
\begin{equation}\label{do}
\Phi(z^{\ast})=\Big(z^{\ast}+z^{\ast}\cos{z^{\ast}\bar{l}_{j}}\Big)\prod_{k=1}^{m}\sin{z^\ast l_{k}}\prod_{k\neq j}\sin{z^\ast \bar{l}_{k}}=0,\,\mbox{ for }z^{\ast}\in\bar{\mathcal{Z}}_{0j},  \mbox{ for } j=1,2,\cdots,m. 
\end{equation}
For $j$ consecutively, we have
$\prod_{k=1}^{m}\sin{z^\ast l_{k}}\prod_{k\neq j}\sin{z^\ast \bar{l}_{k}}\neq0$, and consider the zeros of $z^{\ast}(1+\cos{z^{\ast}\bar{l}_{j}})$ for $z^{\ast}\in\bar{\mathcal{Z}}_{0j}$.
Here, using Lemma \ref{35},
\begin{equation}
\delta(\bar{\mathcal{Z}}_{0j})=\delta(\sin{z^{\ast} \bar{l}_{j}})=\frac{\bar{l}_{j}}{\pi}.\label{d1}
\end{equation}
Moreover, $1+\cos{z^{\ast}\bar{l}_{j}}$ has zeros on $\frac{\mathcal{Z}\pi}{\bar{l}_{j}}$ of multiplicity $2$ merely for odd elements in $\mathbb{Z}$. For odd $\mathbb{Z}$, 
that contradicts~(\ref{do}). Therefore, $\cup_{j=1}^{m}\bar{\mathcal{Z}}_{0j}$ is not a proper subset of $\mathcal{Z}$.
\par
Alternatively,  we consider $z^{\ast}\in\mathcal{Z}_{0j}$, that is, the zeros of $\sin{z  l_{j}}$ for $j=1,2,\cdots,m$.
Combining with~(\ref{3314}),
\begin{eqnarray}\nonumber
&&\Phi(z^{\ast})=\Big(\int_{0}^{l_{j}}\sin{z^{\ast}(l_{j}-x)}\overline{q_{j}(x)}dx+\sin{z l_{j}}\sum_{n=1}^{\infty}\frac{q_{j,n}\overline{q_{j,n}}}{z^{\ast 2}-(\frac{n\pi}{l_{j}})^{2}}\Big)\prod_{k\neq j}\sin{z^\ast l_{j}}\prod_{k=1}^{m}\sin{z^\ast \bar{l}_{k}}\\\nonumber
&&\nonumber+\Big( 2z^{\ast}+z^{\ast}\cos{z^{\ast}l_{j}}+2\sin{z l_{j}}\sum_{n=1}^{\infty}\frac{(-1)^{n}n\pi}{l_{j}}\frac{q_{j,n}}{z^{\ast 2}-(\frac{n\pi}{l_{j}})^{2}}\Big)\prod_{k\neq j}\sin{z^\ast l_{j}}\prod_{k=1}^{m}\sin{z^\ast \bar{l}_{k}},\,j=1,2,\cdots,m.
\end{eqnarray}
Now, we suppose that $\mathcal{Z}_{0j}$ is a proper subset of $\mathcal{Z},\,j=1,2,\cdots,m$, and look for contradiction. That is, for $z^{\ast}\in\mathcal{Z}_{0j}$, we have
\begin{equation}\label{3199}
\Phi(z^{\ast})=0, \mbox{ for }z^{\ast}\in\mathcal{Z}\cap\mathcal{Z}_{0j},
\end{equation}
with density
\begin{equation}
\delta(\mathcal{Z}\cap\mathcal{Z}_{0j})=\delta(\mathcal{Z}_{0j})=\frac{l_{j}}{\pi}.
\end{equation}
Let
\begin{eqnarray}\nonumber
&&G_{j}(z ):=\Big(\int_{0}^{l_{j}}\sin{z (l_{j}-x)}\overline{q_{j}(x)}dx+\sin{z l_{j}}\sum_{n=1}^{\infty}\frac{q_{j,n}\overline{q_{j,n}}}{z^{  2}-(\frac{n\pi}{l_{j}})^{2}}\\\nonumber
&&\nonumber+ 2z +z \cos{z l_{j}}+2\sin{z l_{j}}\sum_{n=1}^{\infty}\frac{(-1)^{n}n\pi}{l_{j}}\frac{q_{j,n}}{z^{  2}-(\frac{n\pi}{l_{j}})^{2}}\Big),\,j=1,2,\cdots,m,
\end{eqnarray}
which is an entire function of type $l_{j}$. Due to rational independence, $\prod_{k=1}^{m}\sin{z^\ast l_{k}}\prod_{k\neq j}\sin{z^\ast \bar{l}_{k}}$ is non-zero, and then $G_{j}(z^{\ast})=0$ on $\mathcal{Z}\cap\mathcal{Z}_{0j}=\mathcal{Z}_{0j}$. Using Cartwright's theory as Lemma \ref{CC}, we deduce 
$$G_{j}(z)\equiv_{\mathbb{C}}0.$$
Then,~(\ref{3314}) becomes
\begin{equation}
\Phi(z)=\sum_{j=1}^{m}\Big(z+z\cos{z\bar{l}_{j}}\Big)\prod_{k=1}^{m}\sin{z l_{k}}\prod_{k\neq j}\sin{z \bar{l}_{k}},
\end{equation}
in which $1+\cos{z^{\ast}\bar{l}_{j}}$ has zeros on $\frac{\mathcal{Z}\pi}{\bar{l}_{j}}$ of multiplicity $2$ merely for odd elements in $\mathbb{Z}$. For odd $\mathbb{Z}$, $\Phi(z)$ has no zero on $\bar{\mathcal{Z}}_{0j}$ that contradicts~(\ref{3199}) again.
The theorem is thus proven.

\end{proof}
Then, we move for some inverse spectral results.
\begin{proposition}\label{P1}
Let the differential system $(\mathcal{D})$ be defined on the extended star graph 
$\overline{T}^{\sigma}$ with same potential function on each edge and with length $\{l_{j},\bar{l}_{j}^{\sigma}\}$ that are rationally independent respectively for $\sigma=1,2$. If $ \Phi^{1}(z)\equiv_{\mathbb{C}}\Phi^{2}(z)$, then $\theta^{1}_{j}=\theta^{2}_{j}$ for $1\leq j\leq m$.
\end{proposition}
\begin{proof}
As assumed in proposition and constructed as in Introduction, we derive from~(\ref{Phi}) that
\begin{eqnarray}\nonumber
&&\Phi^{1}(z)-\Phi^{2}(z)=\sum_{j=1}^{m}z\cos{z l_{j}}\prod_{k\neq j}\sin{z l_{k}}\prod_{k=1}^{m}\sin{z \bar{l}^{1}_{k}}-\sum_{j=1}^{m}z\cos{z l_{j}}\prod_{k\neq j}\sin{z l_{k}}\prod_{k=1}^{m}\sin{z \bar{l}^{2}_{k}}
\\&&-\sum_{j=1}^{m}z\cos{z\bar{l}^{1}_{j}}\prod_{k=1}^{m}\sin{z l_{k}}\prod_{k\neq j}\sin{z \bar{l}^{1}_{k}}+\sum_{j=1}^{m}z\cos{z\bar{l}^{2}_{j}}\prod_{k=1}^{m}\sin{z l_{k}}\prod_{k\neq j}\sin{z \bar{l}^{2}_{k}}\equiv_{\mathbb{C}}0,\label{3.4}
\end{eqnarray}
by proposition assumption.
Given
\begin{equation}
\mathcal{Z}_{0,j}=\big\{z\in\mathbb{C}\big|\,\sin{z l_{j}}=0\big\},
\end{equation}
we plug $\mathcal{Z}_{0,j}$ into~(\ref{3.4}), and obtain consecutively for all $1\leq j\leq m$ and obtain
\begin{equation}\nonumber
z\cos{z l_{j}}\prod_{k\neq j}\sin{z l_{k}}\prod_{k=1}^{m}\sin{z \bar{l}^{1}_{k}}-z\cos{z l_{j}}\prod_{k\neq j}\sin{z l_{k}}\prod_{k=1}^{m}\sin{z \bar{l}^{2}_{k}}
=0.
\end{equation}
Cancelling the zeros of $z\cos{z l_{j}}\prod_{k\neq j}\sin{z l_{k}}$, 
\begin{equation}
\prod_{k=1}^{m}\sin{z \bar{l}^{1}_{k}}-\prod_{k=1}^{m}\sin{z \bar{l}^{2}_{k}}=0,\,z\in\cup_{j=1}^{m}\mathcal{Z}_{0,j}.\label{3.9}
\end{equation}
Thus, zero density on the left of~(\ref{3.9}) 
$$\delta(\prod_{k=1}^{m}\sin{z \bar{l}^{1}_{k}}-\prod_{k=1}^{m}\sin{z \bar{l}^{2}_{k}})=\max\{\frac{\sum_{k=1}^{m}\bar{l}^{1}_{j}}{{\pi}},\frac{\sum_{k=1}^{m}\bar{l}^{2}_{j}}{{\pi}}\}.$$
However, on the right of~(\ref{3.9}), we have the zeros for $\prod_{k\neq j}\sin{z \bar{l}^{1}_{k}}-\prod_{k\neq j}\sin{z \bar{l}^{2}_{k}}$ which contributes
\begin{equation}
\delta(\cup_{j=1}^{m}\mathcal{Z}_{0,j})=\frac{\sum_{j=1}^{m}l_{j}}{{\pi}}
\end{equation}
due to rational independence assumption, and additionally there are the zeros from the common values of $\prod_{k=1}^{m}\sin{z \bar{l}^{1}_{k}}$ and $\prod_{k=1}^{m}\sin{z \bar{l}^{2}_{k}}$ which contributes the density $$\max\{\frac{\sum_{k=1}^{m}\bar{l}^{1}_{j}}{{\pi}},\frac{\sum_{k=1}^{m}\bar{l}^{2}_{j}}{{\pi}}\}.$$ This implies that~(\ref{3.9}) is entirely zero in $\mathbb{C}$. Hence, this implies $\bar{l}_{j}^{1}=\bar{l}_{j}^{2}$ for all $1\leq j\leq m$.
Thus, we prove that the lengths of all extended edges are equal from both sides, and then the uniqueness of $\{\theta_{j}\}_{j=1}^{m}$.

\end{proof}
We remark that all angles between the edges $e^{\sigma}_{j}$ and $e^{\sigma}_{j+1}$, $\sigma=1,2,$ being identical does not imply that $\overline{T}^{1}$ and $\overline{T}^{2}$ are concurrent. One step further, if certain knowledge of $\{l_{j}\}_{j=1}^{m}$ and $\{\bar{l}_{j}\}_{j=1}^{m}$ are given, may we recover the information on $\{q_{j}\}_{j=1}^{m}$?
\begin{theorem}\label{T3.3}
Let the differential system $(\mathcal{D})$ be defined on the star graph 
$\overline{T}^{\sigma}$, $\sigma=1,2,$ with the potential function $q^{\sigma}_{j}$ on each  edge $e_{j}$ with known length $\{l_{j},\bar{l}_{j}\}$ that are rationally independent and $\Phi^{\sigma}(z)$ be the characteristic function defined on $\overline{T}^{\sigma}$. If $ \Phi^{1}(z)=\Phi^{2}(z)$ for $z\in\mathbb{C}$, then $q^{1}_{j}(x)\equiv q^{2}_{j}(x)$ for $1\leq j\leq m$.
\end{theorem}
\begin{proof}
We consider the subtraction from~(\ref{Phi}), and deduce
\begin{eqnarray}\nonumber
&&\Phi^{1}(z)-\Phi^{2}(z)\\
&\nonumber\hspace{-5pt}=\hspace{-5pt}&\sum_{j=1}^{m}\Big(\int_{0}^{l_{j}}\sin{z(l_{j}-x)}\overline{q^{1}_{j}(x)-q^{2}_{j}(x)}dx+\sin{z l_{j}}\sum_{n=1}^{\infty}\frac{q^{1}_{j,n}\overline{q^{1}_{j,n}}-q^{2}_{j,n}\overline{q^{2}_{j,n}}}{z^{2}-(\frac{n\pi}{l_{j}})^{2}}\Big)\prod_{k\neq j}\sin{z l_{k}}\prod_{k=1}^{m}\sin{z \bar{l}_{k}}\\&&
+2\sum_{j=1}^{m}\Big(\sin{z l_{j}}\sum_{n=1}^{\infty}\frac{(-1)^{n}n\pi}{l_{j}}\frac{q^{1}_{j,n}-q^{2}_{j,n}}{z^{2}-(\frac{n\pi}{l_{j}})^{2}}\Big)\prod_{k\neq j}\sin{z l_{k}}\prod_{k=1}^{m}\sin{z \bar{l}_{k}}.\label{3.8}
\end{eqnarray}
Let us plug the zeros of $\sin{z l_{j}}$ into~(\ref{3.8}) and deduce
\begin{eqnarray}\nonumber
\Phi^{1}(z)-\Phi^{2}(z)&\hspace{-5pt}=\hspace{-5pt}&\Big(\int_{0}^{l_{j}}\sin{z(l_{j}-x)}\overline{q^{1}_{j}(x)-q^{2}_{j}(x)}dx+\sin{z l_{j}}\sum_{n=1}^{\infty}\frac{q^{1}_{j,n}\overline{q^{1}_{j,n}}-q^{2}_{j,n}\overline{q^{2}_{j,n}}}{z^{2}-(\frac{n\pi}{l_{j}})^{2}}\Big)\prod_{k\neq j}\sin{z l_{k}}\prod_{k=1}^{m}\sin{z \bar{l}_{k}}\\&&
+2\Big(\sin{z l_{j}}\sum_{n=1}^{\infty}\frac{(-1)^{n}n\pi}{l_{j}}\frac{q^{1}_{j,n}-q^{2}_{j,n}}{z^{2}-(\frac{n\pi}{l_{j}})^{2}}\Big)\prod_{k\neq j}\sin{z l_{k}}\prod_{k=1}^{m}\sin{z \bar{l}_{k}}=0,\,j=1,2,\cdots,m.\label{3233}
\end{eqnarray}
Then, we deduce that
\begin{eqnarray}\nonumber
&&\Big(\int_{0}^{l_{j}}\sin{z(l_{j}-x)}\overline{q^{1}_{j}(x)-q^{2}_{j}(x)}dx+\sin{z l_{j}}\sum_{n=1}^{\infty}\frac{q^{1}_{j,n}\overline{q^{1}_{j,n}}-q^{2}_{j,n}\overline{q^{2}_{j,n}}}{z^{2}-(\frac{n\pi}{l_{j}})^{2}}\Big)\prod_{k\neq j}\sin{z l_{k}}\prod_{k=1}^{m}\sin{z \bar{l}_{k}}\\&&
+2\Big(\sin{z l_{j}}\sum_{n=1}^{\infty}\frac{(-1)^{n}n\pi}{l_{j}}\frac{q^{1}_{j,n}-q^{2}_{j,n}}{z^{2}-(\frac{n\pi}{l_{j}})^{2}}\Big)\prod_{k\neq j}\sin{z l_{k}}\prod_{k=1}^{m}\sin{z \bar{l}_{k}}=0,\,z\in\cup_{j=1}^{m}\mathcal{Z}_{0j}\subset\mathcal{Z}_{0}.\label{323}
\end{eqnarray}
All these three summands in the parentheses are Fourier transforms due to Lemma \ref{LL33} with corresponding zero set for all $j$. Then, we use Cartwright's theory as Lemma \ref{CC} to deduce
\begin{eqnarray}\nonumber
&&\Big(\int_{0}^{l_{j}}\sin{z(l_{j}-x)}\overline{q^{1}_{j}(x)-q^{2}_{j}(x)}dx+\sin{z l_{j}}\sum_{n=1}^{\infty}\frac{q^{1}_{j,n}\overline{q^{1}_{j,n}}-q^{2}_{j,n}\overline{q^{2}_{j,n}}}{z^{2}-(\frac{n\pi}{l_{j}})^{2}}\Big)\prod_{k\neq j}\sin{z l_{k}}\prod_{k=1}^{m}\sin{z \bar{l}_{k}}\\&&
+2\Big(\sin{z l_{j}}\sum_{n=1}^{\infty}\frac{(-1)^{n}n\pi}{l_{j}}\frac{q^{1}_{j,n}-q^{2}_{j,n}}{z^{2}-(\frac{n\pi}{l_{j}})^{2}}\Big)\prod_{k\neq j}\sin{z l_{k}}\prod_{k=1}^{m}\sin{z \bar{l}_{k}}\equiv0,\,z\in\mathbb{C}.\label{3.24}
\end{eqnarray}
Let us plug in the same zero set of $\int_{0}^{l_{j}}\sin{z(l_{j}-x)}\overline{\big(q^{1}_{j}(x)-q^{2}_{j}(x)\big)}dx$ and $\int_{0}^{l_{j}}\sin{z(l_{j}-x)}\big(q^{1}_{j}(x)-q^{2}_{j}(x)\big)dx$ simultaneously
into~(\ref{3.24}) and deduce that
\begin{equation}
q^{1}_{j,n}\overline{q^{1}_{j,n}}-q^{2}_{j,n}\overline{q^{2}_{j,n}}=0,\,n\in\mathbb{N}.
\end{equation}
Back to~(\ref{3.24}) again,
we deduce that $$\Re\{\int_{0}^{l_{j}}\sin{z(l_{j}-x)}\overline{q^{1}_{j}(x)-q^{2}_{j}(x)}dx\}\equiv0.$$
From Cauchy-Riemann equation, we deduce that the entire function $\int_{0}^{l_{j}}\sin{z(l_{j}-x)}\overline{q^{1}_{j}(x)-q^{2}_{j}(x)}dx$ is a constant function, and the sine-transform is zero at $z=0$. 
Thus, we prove that $q^{1}_{j}(x)\equiv q^{2}_{j}(x),\,a.e.,$ on each edge $e_{j}$ from the injection on $L^{2}(0,l_{j})$.

\end{proof}

\begin{acknowledgement}
The author gratefully acknowledges the financial support of the National Science and Technology Council (NSTC) under Grant No. 113-2115-M-131-001. The content of this manuscript does not necessarily reflect the position or the policy of the administration, and no official endorsement should be inferred, neither.
\end{acknowledgement}


\begin{thebibliography}{widest-label}
\bibitem{Alb}Albeverio S., Hryniv R., and Nizhnik L., Inverse spectral problems for nonlocal Sturm-Liouville operators, Inverse Problems, 23, 523--535 (2007).
\bibitem{Alb2}Albeverio S. and Nizhnik L. P., Schr\"{o}dinger operators with nonlocal point interactions, J. Math. Anal. Appl., 332, 884--895 (2007 ).


\bibitem{Belishev}Belishev M. I., Boundary spectral inverse problem on a class of graphs (trees) by the BC method, Inverse Problems, 20, 647--672 (2004).
\bibitem{Boas}Boas R. P., Entire Functions, Academic Press, New York, 1954.
\bibitem{Boman}Boman J. and Kurasov P., Symmetries of quantum graphs and the inverse scattering problem, Adv. Appl. Math., 35, 58--70 (2005).
\bibitem{Bon2}Bondarenko N.P., Buterin S. A., Vasiliev S. V., An inverse spectral problem for Sturm--Liouville operators with frozen argument, J. Math. Anal. Appl., 472, 1028--1041 (2019).
\bibitem{Bon}Bondarenko N. P., A partial inverse problem for the Sturm--Liouville operator on a star--shaped graph, Anal. Math.
Phys., no. 8, 155--168 (2018).
\bibitem{Carlson}Carlson R., Inverse eigenvalue problems on directed graphs, Trans. Amer. Math. Soc., 351, 4069--4088 (1999).
\bibitem{Exner}Exner P.  and Post O., Convergence of spectra of graph-like thin manifolds, J. Geom. Phys., 54, 77--115 (2005).
\bibitem{Freiling}Freiling G. and Yurko V. A., Inverse problems for differential operators on graphs with general matching conditions, Applicable Analysis, 86, no. 6, 653--667 (2007).
\bibitem{Gerasimenko}Gerasimenko N. I. and Pavlov B. S., Scattering problems on non-compact graphs, Theor. Math. Phys., 74, 230--240  (1988).

\bibitem{Kostrykin}Kostrykin V. and Schrader R., Kirchhoff's rule for quantum wires. II. The inverse problem with possible applications to quantum computers, Fortschr. Phys., 48, 703--716 (2000).
\bibitem{Kottos}Kottos T. and Smilansky U., Quantum graphs: A simple model for chaotic scattering, J. Phys. A: Math. Gen., 36, 3501--3524  (2003).

\bibitem{Kuchment1}Kuchment P., Quantum graphs. I. Some basic structures, Waves in Random Media, 14, 107--128 (2004).
\bibitem{Kuchment2}
Berkolaiko G. and Kuchment P.,
Introduction to Quantum Graphs,
Mathematical Surveys and Monographs, Vol. 186,
American Mathematical Society, Providence, 2013.
\bibitem{Kurasov1}Kurasov P. and Novaszek M., Inverse spectral problem for quantum graphs, J. Phys. A: Math. Gen., 38, 4901--4915 (2005).
\bibitem{Kurasov2}Kurasov, P., Spectral Geometry of Graphs, Operator Theory: Advances and Applications,
Volume 293, Birkh\"{a}user, Berlin, 2024.
\bibitem{Levin}Levin B. Ja., Distribution of Zeros of Entire
Functions, revised edition, Translations of Mathematical
Monographs, American Mathematical Society, Providence, 1972.
\bibitem{Levin2}Levin B. Ja., Lectures on Entire Functions,
Translations of Mathematical Monographs, V. 150, AMS, Providence, 1996.

\bibitem{Naimark}Naimark K. and Solomyak M., Eigenvalue estimates for the weighted Laplacian on metric trees, Proc. London Math. Soc., 80, 690--724  (2000).


\bibitem{Nizhnik1}Nizhnik L. P., Inverse eigenvalue problems for nonlocal Sturm--Liouville operators on a star graph, Methods of Functional Analysis and Topology,
Vol.18, no. 1, 68--78 (2012).
\bibitem{Nizhnik2}Nizhnik L. P., Inverse nonlocal Sturm-Liouville problem, Inverse Problems,  26, 125006 (2010). 


\bibitem{Novikov}Novikov S. P., Schr\"{o}dinger operators on graphs and symplectic geometry, E. Bierstone, B. Khesin, A. Khovanskii, and J. Marsden (Eds.), The Arnoldfest, Proceedings of a Conference in Honour of V. I. Arnold for His Sixtieth Birthday, Fields Institute Communications Series, Vol. 24, Amer. Math. Soc., Providence, RI, 397--413 (1999).
\bibitem{Pivovarchik1}Pivovarchik V., Inverse problem for the Sturm-Liouville equation on a simple graph, SIAM J. Math. Anal., 32, 801--819 (2001).
\bibitem{Pivovarchik2}Pivovarchik V., Inverse problem for the Sturm-Liouville equation on a star-shaped graph, Math. Nachr.,  280, no. 13. 14., 1595--1619  (2007).
\bibitem{Pokornyi}Pokornyi Yu. and Pryadiev  V., The qualitative Sturm-Liouville theory on spatial networks, J. Math. Sci. (N.Y.) 119, no. 6, 788--835 (2004).
\bibitem{Tang}Tang S.-H. and Zworski M., Potential scattering on the real line, Department of Mathematics, U.C. Berkeley, http://math.berkeley.edu/~zworski/tz1.pdf.
\bibitem{Titchmarsh}Titchmarsh E. C., The zeros of certain integral functions, Proceedings of The London Mathematical Society, 283--302 (1926).
\bibitem{Yang}Yang, C. -F., Inverse spectral problems for the Sturm--Liouville operator on a d-star graph, J. Math. Anal. Appl., 365, 742--749  (2010).
\bibitem{Yurko1}Yurko V. A., Inverse spectral problems for Sturm-Liouville operators on graphs, Inverse Problems, 21, 1075--1086 (2005).
\bibitem{Yurko2}Yurko V. A., An inverse problem for higher-order differential operators on star-type graphs, Inverse Problems, 23, 893--903 (2007).
\bibitem{Shieh}Wang, Y. and Shieh, C. -T., Inverse problems for Sturm--Liouville operators on a star-shaped graph with mixed spectral data, Applicable Analysis,  Vol. 99, Issue 14,1--10 (2019).
\bibitem{Shieh2}Shieh, C.-T., Tsai, T.-M., and  Wu, M.-N., Partial Inverse Spectral Problems for Sturm-Liouville Operators with Frozen Arguments on a Star-Shaped Graph, Results Math., article number 80, 195 (2025).



\end{thebibliography}
\end{document}